\documentclass[11pt]{article}
\usepackage[table]{xcolor}
\usepackage[utf8]{inputenc}
\usepackage[english]{babel}
\usepackage{amsmath}
\usepackage{amsfonts}
\usepackage{amssymb}
\usepackage{amsthm}
\usepackage[basic]{complexity}
\usepackage{tikz}
\usepackage{hyphenat}
\usepackage{graphicx}
\usepackage{nicefrac}
\usepackage{booktabs}
\usepackage{dsfont}
\usepackage{bm}
\usepackage{setspace}

\usepackage{libertine}

\usepackage[authoryear]{natbib}

\usepackage{geometry}
\usepackage{tikz}
\usetikzlibrary{automata, positioning, calc, shapes, arrows, fit, patterns, patterns.meta}

\usepackage{tabularx}

\usepackage[unicode=true, bookmarks=false, breaklinks=true, pdfborder={0 0 1}, colorlinks=true]{hyperref}
\hypersetup{linkcolor=my-blue-dark, citecolor=my-blue-dark, filecolor=my-blue, urlcolor=my-blue-dark}

\theoremstyle{plain}
\newtheorem{theorem}{Theorem}

\newtheorem{definition}{Definition}

\usepackage{xcolor}

\definecolor{my-red}{HTML}{C62828}
\definecolor{my-red-light}{HTML}{E57373}
\definecolor{my-red-verylight}{HTML}{FFCDD2}
\definecolor{my-red-dark}{HTML}{B71C1C}

\definecolor{my-pink}{HTML}{EC407A}
\definecolor{my-pink-light}{HTML}{F48FB1}
\definecolor{my-pink-verylight}{HTML}{F8BBD0}
\definecolor{my-pink-dark}{HTML}{880E4F}

\definecolor{my-purple}{HTML}{8E24AA}
\definecolor{my-purple-light}{HTML}{BA68C8}
\definecolor{my-purple-verylight}{HTML}{e5cefc}
\definecolor{my-purple-dark}{HTML}{6A1B9A}

\definecolor{my-indigo}{HTML}{3949AB}
\definecolor{my-indigo-light}{HTML}{7986CB}
\definecolor{my-indigo-verylight}{HTML}{9FA8DA}
\definecolor{my-indigo-dark}{HTML}{1A237E}

\definecolor{my-blue}{HTML}{1E88E5}
\definecolor{my-blue-light}{HTML}{64B5F6}
\definecolor{my-blue-verylight}{HTML}{B3E5FC}
\definecolor{my-blue-dark}{HTML}{0D47A1}

\definecolor{my-cyan}{HTML}{00BCD4}
\definecolor{my-cyan-light}{HTML}{4DD0E1}
\definecolor{my-cyan-verylight}{HTML}{80DEEA}
\definecolor{my-cyan-dark}{HTML}{0097A7}

\definecolor{my-teal}{HTML}{009688}
\definecolor{my-teal-light}{HTML}{4DB6AC}
\definecolor{my-teal-verylight}{HTML}{B2DFDB}
\definecolor{my-teal-dark}{HTML}{00695C}

\definecolor{my-green}{HTML}{39ac39}
\definecolor{my-green-light}{HTML}{8cd98c}
\definecolor{my-green-verylight}{HTML}{b3e6b3}
\definecolor{my-green-dark}{HTML}{339933}

\definecolor{my-grass}{HTML}{689F38}
\definecolor{my-grass-light}{HTML}{8BC34A}
\definecolor{my-grass-verylight}{HTML}{AED581}
\definecolor{my-grass-dark}{HTML}{33691E}

\definecolor{my-lime}{HTML}{CDDC39}
\definecolor{my-lime-light}{HTML}{DCE775}
\definecolor{my-lime-verylight}{HTML}{E6EE9C}
\definecolor{my-lime-dark}{HTML}{AFB42B}

\definecolor{my-yellow}{HTML}{fffc29}
\definecolor{my-yellow-light}{HTML}{fffd7a}
\definecolor{my-yellow-verylight}{HTML}{fefdbb}
\definecolor{my-yellow-dark}{HTML}{FFD600}

\definecolor{my-orange}{HTML}{FF8F00}
\definecolor{my-orange-light}{HTML}{FFC107}
\definecolor{my-orange-verylight}{HTML}{ffe5a4}
\definecolor{my-orange-verylight}{HTML}{ffe6bb}
\definecolor{my-orange-dark}{HTML}{FF6F00}

\definecolor{my-brown}{HTML}{6D4C41}
\definecolor{my-brown-light}{HTML}{795548}
\definecolor{my-brown-verylight}{HTML}{BCAAA4}
\definecolor{my-brown-dark}{HTML}{3E2723}

\definecolor{my-gray}{HTML}{616161}
\definecolor{my-gray-light}{HTML}{9E9E9E}
\definecolor{my-gray-verylight}{HTML}{f0f0f0}
\definecolor{my-gray-veryverylight}{HTML}{f6f6f6}
\definecolor{my-gray-dark}{HTML}{424242}

\definecolor{my-steel}{HTML}{546E7A}
\definecolor{my-steel-light}{HTML}{78909C}
\definecolor{my-steel-verylight}{HTML}{B0BEC5}
\definecolor{my-steel-dark}{HTML}{37474F}

\definecolor{ColorBlind1}{HTML}{D81B60}
\definecolor{ColorBlind2}{HTML}{1E88E5}
\definecolor{ColorBlind3}{HTML}{FFC107}
\definecolor{ColorBlind4}{HTML}{004D40}

\definecolor{MyOrange}{rgb}{1, 0.5, 0}
\definecolor{MyLightOrange}{rgb}{1, 0.9, 0.7}
\definecolor{MyGrey}{rgb}{0.3, 0.3, 0.3}
\definecolor{MyLightGrey}{rgb}{0.9, 0.9, 0.9}
\definecolor{MyGreen}{rgb}{0, 0.6, 0}
\definecolor{MyBlue}{rgb}{0, 0, 0.6}
\definecolor{MyLightBlue}{rgb}{0.7, 0.8, 1}
\definecolor{MyRed}{rgb}{0.7, 0, 0}
\definecolor{MyLightRed}{rgb}{1, 0.7, 0.7}
\definecolor{MyLightYellow}{HTML}{f9f6ec}

\usepackage{framed}

\renewenvironment{leftbar}[1][\hsize]
{%
	\MakeFramed{\hsize#1\advance\hsize-\width\FrameRestore}%
}
{\endMakeFramed}

\usepackage{xparse}
\NewDocumentEnvironment{illustration}{o}{%
	\begin{leftbar}\noindent{\bfseries Illustration\IfNoValueTF{#1}{}{ \normalfont(#1)\bfseries}.}%
}
	{\end{leftbar}}

\usepackage{pifont}

\usepackage{eurosym}
\usepackage{textcomp}

\newcommand{\tuple}[1]{\left\langle #1 \right\rangle}

\renewcommand{\phi}{\varphi}

\newcommand{\agentSet}{\mathcal{N}}
\newcommand{\projSet}{\mathcal{P}}
\newcommand{\allocSet}{\textsc{Feas}}

\newcommand{\profile}{\boldsymbol{A}}

\newcommand{\satisfaction}{\mathit{sat}}

\newcommand{\cardSatisfaction}{\mathit{sat}^{\mathit{card}}}
\newcommand{\costSatisfaction}{\mathit{sat}^{\mathit{cost}}}

\newcommand{\share}{\mathit{share}}

\usepackage{todonotes}

\author{Jan Maly\\
	DBAI, Institute of Logic and Computation, TU Wien, Vienna, Austria}
\title{The core of an approval-based PB instance can be empty for nearly all cost-based 
	satisfaction functions and for the share} 

\begin{document}
\maketitle

\abstract{
The core is a strong fairness notion in multiwinner voting 
and participatory budgeting (PB). 
It is known that the core can be empty if we consider cardinal 
utilities, but it is not known whether it is always satisfiable 
with approval-ballots. In this short note, I show that 
in approval-based PB the core can be empty for nearly all satisfaction functions
that are based on the cost of a project.
In particular, I show that the core can be empty for the cost satisfaction
function, satisfaction functions based on diminishing marginal returns
and the share. 
However, it remains open whether the core can be empty for the cardinality satisfaction 
function.
}

\section{Introduction}

Proportionality or fairness in (approval-based) multiwinner voting \citep{LaSk23}
and participatory budgeting (PB) \citep{Survey} has been one of the
most active research areas in (Computational) Social Choice in recent years.
In this time, many different fairness notions have been introduced, some which can be 
satisfied, like EJR \citep{ABCEFW17}, and some which are known to not always be satisfiable,
like laminar proportionality \citep{PeSk20}. 
There is, however, one important proportionality axiom for which it is unknown whether it 
is satisfiable or not, the core.

The core was introduced in multiwinner voting by \citet{ABCEFW17}.
Intuitively, we say that a committee is in the core if no group of voters $N$
can improve their outcome by `leaving' the election with their share of the seats or budget.
The question whether there always exists a committee that is in the core 
is considered one of the central open questions of the multiwinner voting literature \citep{LaSk23}.

In recent years, there has been a strong push to extend the theory of proportionality 
from multiwinner voting to PB (see the recent survey by \citet{Survey} for a detailed account).
PB is a generalization of multiwinner voting in
which the candidates, usually called projects in PB, can have different weights.
The core for PB was first introduce by \citet{FGM16} for divisible projects.
The definition for indivisible projects, which we are concerned with in this note, was introduced
by \citet{FMS18}, who also showed that the core can be empty when
voters are allowed to submit cardinal ballots.
\citet{PPS21} improved this result, by showing that it even holds for the unit-cost case.
However, both counter-examples rely on voters having preferences that constitute a 
type of Condorcet cycle, which cannot occur if voters only submit approval ballots.
Consequently, the question whether the core is always non-empty has remained open
for approval-based PB. 

In this note, I answer the question negatively for a large class of natural
approval-based satisfaction functions. Approval-based satisfaction functions
are, essentially, different ways of interpreting approval preferences in PB.
In the past, a hand full of different satisfaction function have been considered 
in the literature, with the cost satisfaction function and the cardinality satisfaction
function being the two most prominent one (see again the survey by \citet{Survey} 
for a full list of which papers use which satisfaction function).
Recently \citet{BFLMP23} introduced a general 
framework for reasoning about different satisfaction functions in PB, which I will
follow here. Using this framework, I will show that the core can be empty for a large and 
natural class of satisfaction functions. Crucially, this class contains the
cost satisfaction function, which is the de-facto standard in real world applications,
in particular in elections where proportional voting rules are used.
However, it does not include the cardinality satisfaction function.

This note is organized as follows: First I will introduce the setting of 
approval-based PB and define the core. Then, I will show that the core 
can be empty for the cost satisfaction function. This will both 
serve as an illustration of the proof idea and also give a concrete example
for the most important satisfaction function. Afterwards, I present a general result
that covers essentially all natural satisfaction functions for which the satisfaction
grows with the cost of a project. Finally, I show that the presented proof also 
works for a satisfaction function called the share, which is not covered by the general result.

\section{Preliminaries}

Let us first introduce the formal framework for approval-based PB. For this, we follow the survey of 
\citet{Survey}.	

A PB instance is a tuple of three elements  $I = \tuple{\projSet, c, b}$ called an \emph{instance} where $\projSet = \{p_1, \ldots, p_m\}$ is the \emph{set of projects}; $c: \projSet \rightarrow \mathbb{R}_{> 0}$ is the \emph{cost function}, associating every project $p \in \projSet$ with its cost $c(p) \in \mathbb{R}_{> 0}$; and $b \in \mathbb{R}_{> 0}$ is the \emph{budget limit}. We assume that all projects are feasible $p \in \projSet$, \textit{i.e.}, that $c(p)$ for all $p \in \projSet$. For any subset of projects $P \subseteq \projSet$, we denote by $c(P)$ its total cost $\sum_{p \in P} c(p)$. 
	An instance $I = \tuple{\projSet, c, b}$ is said to have \emph{unit costs} if for every project $p \in \projSet$, we have $c(p) = 1$ and $b \in \mathbb{N}_{> 0}$. These instances are especially interesting because they correspond to multi-winner elections \citep{LaSk23}.
	
	Let $\agentSet = \{1, \ldots, n\}$ be the \emph{set of voters} involved in the PB process. When facing an instance $I = \tuple{\projSet, c, b}$, they are asked to submit their preferences over the projects in $\projSet$. In this note we assume that they are using approval ballots. For agent $i \in \agentSet$, their approval ballot $A_i \subseteq \projSet$ is a subset of $\projSet$, where $p \in A_i$ indicates that agent $i$ approves of project $p$. 

The outcome of a PB instance $I = \tuple{\projSet, c, b}$ is a \emph{budget allocation} $\pi \subseteq \projSet$ such that $c(\pi) \leq b$. We will denote by $\allocSet(I)$ the set of all \emph{feasible} budget allocations for instance~$I$, defined as $\allocSet(I) = \{\pi \subseteq \projSet \mid c(\pi) \leq b\}$. A budget allocation $\pi \in \allocSet(I)$ is \emph{exhaustive} if there is no $p \in \projSet \setminus \pi$ such that $c(\pi \cup \{p\}) \leq b$.
	
	
	When it comes to approval ballots, there is no obvious way to define a measure of the satisfaction of a voter. \citet{BFLMP23} introduced the concept of \emph{approval-based satisfaction functions}, which are functions translating a budget allocation into a satisfaction level for the agents, given their approval ballots. Let us provide their definition.
	
	\begin{definition}[Approval-Based Satisfaction Functions]
		Given an instance $I = \tuple{\projSet, c, b}$ and a profile $\profile$, an \emph{(approval-based) satisfaction function} is a mapping $\satisfaction: 2^\projSet \rightarrow \mathbb{R}_{\geq 0}$ satisfying the following two conditions:
		\begin{itemize}
			\item $\satisfaction(P) \geq \satisfaction(P')$ for all $P, P' \subseteq \projSet$ such that $P \supseteq P'$: the satisfaction is inclusion-monotonic;
			\item $\satisfaction(P) = 0$ if and only if $P = \emptyset$: the satisfaction is zero only for the empty set.
		\end{itemize}
		
		\noindent The satisfaction of agent $i \in \agentSet$ for a budget allocation $\pi \in \allocSet(I)$ is defined as:
		\[\satisfaction_i(\pi) = \satisfaction(\{p \in \pi \mid A_i(p) = 1\}).\]
	\end{definition}
	
	\noindent Note that in contrast to the case of cardinal ballots, satisfaction functions are not generally assumed to be additive. 
	
	Several satisfaction functions have been introduced in the literature, we define them below.
	
	\begin{itemize}
		\item \textbf{Cardinality Satisfaction Function} \citep{TaFa19}: measures the satisfaction of the voters as the number of selected and approved projects:
		\[\cardSatisfaction(P) = |P|.\]
		\item \textbf{Cost Satisfaction Function} \citep{TaFa19}: measures the satisfaction of the voters as the cost of the selected and approved projects:
		\[\costSatisfaction(P) = c(P).\]
		Note that with indivisible projects, this is equivalent to the \emph{overlap satisfaction function} of \citet{GKSA19}.
		\item \textbf{Chamberlin-Courant Satisfaction Function} \citep{TaFa19}: measures the satisfaction of the voters as being 1 if at least one approved project was selected, and 0 otherwise:
		\[\satisfaction^{CC}(P) = \mathds{1}_{P \neq \emptyset}.\]
		\item \textbf{Share} \citep{LMR21}: measures the resources the decision maker used to satisfy the voters:
		\[\satisfaction^{\share}(P) = \sum_{p \in P} \frac{c(p)}{|\{i \in \agentSet \mid A_i(p) = 1\}|}.\]
		It is important to keep in mind that the share has not been introduced as a satisfaction function but can still be interpreted as one (while being cautious as to how to use it).
		\item \textbf{Square Root and Log Satisfaction Functions} \citep{BFLMP23}: measures the satisfaction of the voters as (marginally) diminishing when the cost of a project increases:
		\[\satisfaction^{\sum\ln}(P) = \sum_{p \in P} \ln(1 + c(p))  \qquad\qquad  \satisfaction^{\sum\sqrt{~}}(P) = \sum_{p \in P} \sqrt{c(p)}.\]
		Note that we could also consider satisfaction functions that implement global marginal diminishing satisfaction:
		\[\satisfaction^{\ln}(P) = \ln(1 + c(P))  \qquad\qquad  \satisfaction^{\sqrt{~}}(P) = \sqrt{c(P)}.\]

	\end{itemize}

To finish our preliminaries, let us introduce the core. Observe that it is always defined 
with respect to a specific satisfaction function $\satisfaction$.
	
	\begin{definition}[The Core of PB with Approval Ballots]
		Given an instance $I = \tuple{\projSet, c, b}$, a profile $\profile$ of approval ballots and a satisfaction function $\satisfaction$, a budget allocation $\pi \in \allocSet(I)$ is \emph{in the core[$\satisfaction$]} of $I$ for $\satisfaction$ if there is no group of voters $N \subseteq \agentSet$ and subset of projects $P \subseteq \projSet$ such that $\nicefrac{|N|}{n} \geq \nicefrac{c(P)}{b}$ and for all voter $i^\star \in N$ we have:
		\[\satisfaction_{i^\star}(\pi) < \satisfaction_{i^\star}(P).\]
If such a group of voters $N$ and set of projects $P$ does exist, we say that the voters in $N$
can deviate to $P$.
	\end{definition}

\section{Cost-satisfaction function}

First, let us show that the core can be empty for the cost satisfaction function.

\begin{theorem}\label{thm:1}
There are PB instances for which the core$[\costSatisfaction]$ is empty.
\end{theorem}

\begin{proof}
Consider the following PB-instance with three voters $1$, $2$ and $3$,
nine projects in three groups, the joint projects $p_{12}$, $p_{13}$ and $p_{23}$,
the large personal projects $\ell_1$, $\ell_2$ and $\ell_3$ and the small personal projects
$s_1$, $s_2$ and $s_3$. The approvals are as follows: Voter $1$ approves
$p_{12}$, $p_{13}$, $\ell_1$ and $s_1$, voter $2$ approves $p_{12}$, $p_{23}$, $\ell_2$ and $s_2$
and voter $3$ approves $p_{13}$, $p_{23}$, $\ell_3$ and $s_3$. In other words, each voter approves 
all projects that have their name in the subscript.

The costs of the projects are as follows: $c(p_{12}) = c(p_{23}) = c(p_{13}) = 8$ while
$c(\ell_i) = 5$ and $c(s_i) = 2$, for all $i \in 1,2,3$.
For the budget, we have $b = 15$.

We claim that no feasible allocation $\pi$ can be in the core. First assume $\pi$
does not contain any of the projects $p_{12},p_{13},p_{23}$. Then for all voters $i$,
we have $\costSatisfaction_i(\pi) \leq 5 + 2 = 7$.
However, any set of two voters, say w.l.o.g.\ $1$ and $2$, 
deserves $\nicefrac{2}{3} b = 10$, hence they can deviate to
$\pi^* = \{p_{12}\}$. Then $\costSatisfaction_1(\pi^*) = 8 > 7 = \costSatisfaction_1(\pi)$
and the same for $2$.

It follows that $\pi$ contains at least one of the projects $p_{12}, p_{13}, p_{23}$.
Due to the symmetry of the instance, we can assume w.l.o.g.\ that $p_{12} \in \pi$.
Observe that $b = 15 < 16= c(p_{12}) + c(p_{13}) =  c(p_{12}) + c(p_{23})$. Hence $p_{13}$
and $p_{23}$ cannot be in $\pi$.
Next, assume that $\pi$ does not contain $\ell_3$. Then, as $\pi$ does not contain $\ell_3$, $p_{13}$
and $p_{23}$, we know that $\costSatisfaction_3(\pi) \leq 2$.
However, $\nicefrac{1}{3}b = 5$. Therefor, $3$ can deviate to $\pi^* = \{\ell_3\}$,
for which we have $\costSatisfaction_3(\pi^*) = 5$

It follows that we must have $p_{12} \in \pi$ and $\ell_3 \in \pi$.
As $b - c(p_{12}) - c(\ell_3) = 15 - 8 - 5 = 2$,
we can now only fit one more project from $s_1, s_2, s_3$. This means that either $s_1$ or $s_2$
is not in $\pi$. Again w.l.o.g.\ assume $s_1$ is not in $\pi$. It follows that
$\costSatisfaction_1(\pi) = 8$ and $\costSatisfaction_3(\pi) \leq 7$.
However, as before, $1$ and $3$ together deserve $\nicefrac{2}{3}b = 10$ units of money.
This means they can deviate to $\pi^* = \{p_{13}, s_1\}$. Then $\costSatisfaction_1(\pi^*) =
8+2 > \costSatisfaction_1(\pi)$ and $\costSatisfaction_3(\pi^*) = 8 > \costSatisfaction_3(\pi)$.
A contradiction.
\end{proof}

\section{A generalization to large class of satisfaction functions}

Let us know provide a more general result. For the proof to work, we need
some constraints on the satisfaction function.

\begin{theorem}\label{thm:2}
Let $\satisfaction$ be a satisfaction function such that
\begin{enumerate}
\item $P \subsetneq P'$ implies $\satisfaction(P) < \satisfaction(P')$,
\item $c(p) < c(q)$ implies $\satisfaction(p) < \satisfaction(q)$ for all $p,q \in \projSet$
\end{enumerate}
and such that there exist
$b$ and $\epsilon$ for which 
\begin{enumerate}
\setcounter{enumi}{2}
\item $\nicefrac{2b}{3}- \epsilon > \nicefrac{b}{2}$,
\item For any projects with $c(p_\epsilon) = \epsilon$, $c(p_{\nicefrac{b}{3}}) = \nicefrac{b}{3}$
and $c(p_{\nicefrac{2b}{3}- \epsilon}) = \nicefrac{2b}{3}- \epsilon$ we have
\[
\satisfaction(\{p_{\nicefrac{b}{3}}, p_{\epsilon}\})<
\satisfaction(p_{\nicefrac{2b}{3}- \epsilon}) 
\]
\end{enumerate}
Then, there are PB instances for which the core$[\satisfaction]$ is empty.
\end{theorem}

Conditions 1 and 2 are very natural monotonicity conditions.
Condition 1, which is a basic strict monotonicity notion,
is satisfied by basically any satisfaction function that is not based on some type of minmax
procedure like the Chamberlin-Courant satisfaction function.
Condition 2 ensures that more expensive projects provide higher satisfaction.
This is a more demanding property and is, for example, not satisfied
by the cardinality satisfaction function or the share. In the next section,
we will see that the proof idea shown this note also works for the share.
Unfortunately, this does not seem to be the case for the cardinality satisfaction
function.

Condition 3 and 4 are more technical. Condition 3 is there to ensure 
that $\epsilon$ is small enough in relation to $b$ and could be rewritten as
$\epsilon < \nicefrac{b}{6}$. However, the formulation above will be more useful in the proof.
Condition 4 is not a very natural requirement, but, if $b$ is large enough and
$\epsilon$ is small enough, it should hold for essentially any
natural satisfaction function that satisfies Condition 2.
Consider, for example, $b = 9999$ and $\epsilon =0.5$. Then $\satisfaction^{\sum\ln}$
satisfies Condition 4, as $\ln(1+3333) + \ln(1+0.5) \approx 8.5$, while
$\ln(1 + 6666 -1) \approx 8.8$. It is straightforward to check that the
same $b$ and $\epsilon$ also work for $\satisfaction^{\sum\sqrt{~}}$,
$\satisfaction^{\ln}$, $\satisfaction^{\sqrt{~}}$.

\begin{proof}
We again consider a PB-instance with three voters $1$, $2$ and $3$,
and nine projects in the same three groups, $p_{12}$, $p_{13}$ and $p_{23}$;
$\ell_1$, $\ell_2$ and $\ell_3$; $s_1$, $s_2$ and $s_3$.
The approvals are as before: Voter $1$ approves
$p_{12}$, $p_{13}$, $\ell_1$, and $s_1$, $2$ approves $p_{12}$, $p_{23}$, $\ell_2$ and $s_2$,
$3$ approves $p_{13}$, $p_{23}$, $\ell_3$ and $s_3$. In other words, each voter approves 
all projects that have their name in the subscript.

In contrast to the first proof, the costs of the projects are, however, changed:
$c(p_{12}) = c(p_{23}) = c(p_{13})
= \nicefrac{2b}{3}- \epsilon$ while $c(\ell_i) = \nicefrac{b}{3}$ and $c(s_i) = \epsilon$ ,
for all $i \in 1,2,3$. For the budget, we have $b$.
Now, we claim that basically the same argumentation as in the proof of Theorem~\ref{thm:1}
still works. Let us go through 
the steps one by one.

We claim that no feasible allocation $\pi$ can be in the core. 
Let us assume for the sake of a contradiction that $\pi$ is in the core.
First assume additionally that $\pi$
does not contain any of the projects $p_{12},p_{13},p_{23}$. Then for all voters $i$,
we have $A_i \cap \pi \subseteq \{\ell_i, s_i\}$.
However, any set of two voters, say w.l.o.g.\ $1$ and $2$, 
deserves $\nicefrac{2}{3} b$, hence they can deviate to
$\pi^* = \{p_{12}\}$. However, then $\satisfaction_1(\pi^*) =
\satisfaction(p_{12})$, while, by Condition 1,
$\satisfaction_1(\pi) \leq \satisfaction(\{\ell_1,s_1\})$,
which implies by Condition 4 that $\satisfaction_1(\pi^*) > \satisfaction_1(\pi)$.
The same holds for $2$. This contradicts the assumption that $\pi$ is in the core.

It follows that $\pi$ contains at least one of the projects $p_{12}, p_{13}, p_{23}$.
Due to the symmetry of the instance, we can assume w.l.o.g.\ that $p_{12} \in \pi$.
Recall that $b < 2c(p_{12})$ by Condition 3. Hence $p_{13}$ and $p_{23}$ cannot be in $\pi$.
Next, assume that $\pi$ does not contain $\ell_3$. Then, as $\pi$ does not contain $\ell_3$, $p_{13}$
and $p_{23}$, we know that $A_3 \cap \pi \subseteq \{s_i\}$.
However, $3$ deserve $\nicefrac{1}{3}b$ units of money and $c(\ell_i) = \nicefrac{b}{3}$.
Therefor, $3$ can deviate to $\pi^* = \{\ell_3\}$,
Now, as $c(\ell_3) > c(s_3)$, we have $\satisfaction_3(\pi^*) > \satisfaction_3(\{s_3\}) \geq 
\satisfaction_3(\pi)$ by Condition 2. This again contradicts the assumption that $\pi$ is in the core.

It follows that we must have $p_{12} \in \pi$ and $\ell_3 \in \pi$.
As $b - c(p_{12}) - c(\ell_3) = b - (\nicefrac{2}{3}b - \epsilon) - \nicefrac{b}{3} = \epsilon$,
we can now only fit one more project from $s_1, s_2, s_3$. This means that either $s_1$ or $s_2$
is not in $\pi$. Again w.l.o.g.\ assume $s_1$ is not in $\pi$. It follows that
$A_1 \cup \pi = p_{12}$ and $A_3 \cap \pi \subseteq \{\ell_3, s_3\}$.
However, $1$ and $3$ together deserve $\nicefrac{2}{3}b$ units of money.
This means they can deviate to $\pi^* = \{p_{13}, s_1\}$, as $c(\pi^*) = \nicefrac{2}{3}b$.
Then $\satisfaction_1(\pi^*) = \satisfaction(\{p_{13}, s_1\}) > \satisfaction(\{p_{13}\}) =
\satisfaction_1(\pi)$ by Condition 1. Moreover, we have 
$\satisfaction_3(\pi^*) = \satisfaction(p_{13})$
which is by Condition 4 larger than $\satisfaction(\{\ell_3,s_3\}) \geq \satisfaction_3(\pi)$.
A contradiction.
\end{proof}

\section{The share}

As mentioned above, the share does not satisfy the conditions of Theorem~\ref{thm:2}.
However, the same proof still works with slightly different costs and a slight variation
on the argumentation. In the following, we write $\share(\pi)$ instead of
$\satisfaction^{\share}(\pi)$ to improve the readability.

\begin{theorem}\label{thm:3}
There are PB instances for which the core$[\share]$ is empty.
\end{theorem}

\begin{proof}
The voters, projects and approvals are the same as in Theorem~\ref{thm:1} and~\ref{thm:2},
i.e., there are three voters $1$, $2$ and $3$,
nine projects in three groups, the joint projects $p_{12}$, $p_{13}$ and $p_{23}$,
the large personal projects $\ell_1$, $\ell_2$ and $\ell_3$ and the small personal projects
$s_1$, $s_2$ and $s_3$. The approvals are as follows: Voter $1$ approves
$p_{12}$, $p_{13}$, $\ell_1$ and $s_1$, voter $2$ approves $p_{12}$, $p_{23}$, $\ell_2$ and $s_2$
and voter $3$ approves $p_{13}$, $p_{23}$, $\ell_3$ and $s_3$. In other words, each voter approves 
all projects that have their name in the subscript.

The costs of the projects is slightly different than before:
$c(p_{12}) = c(p_{23}) = c(p_{13}) = 11$ while
$c(\ell_i) = 3$ and $c(s_i) = 2$, for all $i \in 1,2,3$.
For the budget, we have $b = 21$.

As before, we claim that no feasible allocation $\pi$ can be in the core. First assume $\pi$
does not contain any of the projects $p_{12},p_{13},p_{23}$. Then for all voters $i$,
we have $\share_i(\pi) \leq \share_i(\{\ell_i,s_i\} = 5$.
However, any set of two voters, say w.l.o.g.\ $1$ and $2$, 
deserves $\nicefrac{2}{3} b = 14$, hence they can deviate to
$\pi^* = \{p_{12}\}$. Then $\share_1(\pi^*) = \nicefrac{11}{2} = 5.5 > 5 = \share_1(\pi)$
and the same for $2$.

It follows that $\pi$ contains at least one of the projects $p_{12}, p_{13}, p_{23}$.
Due to the symmetry of the instance, we can assume w.l.o.g.\ that $p_{12} \in \pi$.
Observe that $b = 21 < 22= c(p_{12}) + c(p_{13}) =  c(p_{12}) + c(p_{23})$. Hence, $p_{13}$
and $p_{23}$ cannot be in $\pi$.
Next, assume that $\pi$ contains at most one of $\ell_3$ or $s_3$.
Then, as $\pi$ also does not contain $p_{13}$ and $p_{23}$,
we know that $\share_3(\pi) \leq \share_3(\ell_3) = 3$.
However, $\nicefrac{1}{3}b = 7$. Therefor, $3$ can deviate to $\pi^* = \{\ell_3, s_3\}$,
for which we have $\share_3(\pi^*) = 5$

It follows that we must have $p_{12} \in \pi$ and both $\ell_3$ and $s_3$ in $\pi$.
As $b - c(p_{12}) - c(\ell_3) - c(s_3)= 21 - 11 - 3 -2 = 5$,
we can now only fit either $\ell_1$ or  $\ell_2$ in $\pi$, but not both.
Again w.l.o.g.\ assume $\ell_1$ is not in $\pi$. It follows that
$\share_1(\pi) \leq \share_1(\{p_{12},s_1\}) = 7.5$ and
$\share_3(\pi) \leq \share_3(\{\ell_3,s_3\})$.
However, as before, $1$ and $3$ together deserve $\nicefrac{2}{3}b = 14$ units of money.
This means they can deviate to $\pi^* = \{p_{13}, \ell_1\}$. Then $\share_1(\pi^*) =
8.5 > \share_1(\pi)$ and $\share_3(\pi^*) = 5.5 > \share_3(\pi)$.
A contradiction.
\end{proof}

\subsection*{Acknowledgments}

First of all, I would like to thank Adrian Haret, Sophie Klumper, Guido Schäfer and 
Simon Rey. The work and discussions with them inspired the example underlying this note.
Secondly, I would like to thank Julian Chingoma and, again, Simon Rey, who helped me check the 
correctness of Theorem 1 and Theorem 3. If there are nevertheless any mistakes in these
proofs, they are all my fault. Finally, I would like to thank the Austrian Science Fund (FWF),
which funded this research under grant number J4581.

	\bibliographystyle{ACM-Reference-Format}
	\bibliography{PB}
	
\end{document}